\documentclass[11pt, letterpaper, onecolumn]{IEEEtran}

\usepackage{graphics} 
\usepackage{epsfig} 
\usepackage{times} 
\usepackage{amsmath} 
\usepackage{amsthm}
\usepackage{amssymb}
\usepackage{tikz}
\usetikzlibrary{graphs}
\usepackage[colorlinks=false, urlcolor=blue, pdfborder={0 0 0}]{hyperref}
\usepackage{enumerate}
\usepackage{color}
\usepackage[linesnumbered,boxed,ruled,commentsnumbered]{algorithm2e}
\usepackage{comment}
\usepackage{cite}
\usepackage{tabularx}
\usepackage{booktabs}
\usepackage{subcaption}
\usepackage{enumitem}
\usepackage{booktabs,arydshln}

\makeatletter
\def\adl@drawiv#1#2#3{%
        \hskip.5\tabcolsep
        \xleaders#3{#2.5\@tempdimb #1{1}#2.5\@tempdimb}%
                #2\z@ plus1fil minus1fil\relax
        \hskip.5\tabcolsep}
\newcommand{\cdashlinelr}[1]{%
  \noalign{\vskip\aboverulesep
            \global\let\@dashdrawstore\adl@draw
            \global\let\adl@draw\adl@drawiv}
  \cdashline{#1}
  \noalign{\global\let\adl@draw\@dashdrawstore
            \vskip\belowrulesep}}
\makeatother

\newtheorem{mydef}{Definition}
\newtheorem{mythm}{Theorem}

\newtheorem{remark}{Remark}

\newcommand{\change}[1]{\textcolor{black}{#1}}

\allowdisplaybreaks
\providecommand{\keywords}[1]
{
  \small	
  \textbf{\textit{Index Terms---}} #1
}

\title{Differentiable Safe Controller Design through Control Barrier Functions}

\author{Shuo Yang$^*$, Shaoru Chen$^*$, Victor M. Preciado, and Rahul Mangharam
\thanks{Shuo Yang, Shaoru Chen, Victor M. Preciado, Rahul Mangharam are with the Department
of Electrical and Systems Engineering, University of Pennsylvania, Philadelphia, PA 19104, USA. Email: \{yangs1, srchen, preciado, rahulm\}@seas.upenn.edu.}
\thanks{$^*$Authors contributed equally.}
}

\begin{document}
\maketitle

\begin{abstract}
Learning-based controllers, such as neural network (NN) controllers, can show high empirical performance but lack formal safety guarantees. To address this issue, control barrier functions (CBFs) have been applied as a safety filter to monitor and modify the outputs of learning-based controllers in order to guarantee the safety of the closed-loop system. However, such modification can be myopic with unpredictable long-term effects. In this work, we propose a safe-by-construction NN controller which employs differentiable CBF-based safety layers and relies on a set-theoretic parameterization. We compare the performance and computational complexity of the proposed controller and an alternative projection-based safe NN controller in learning-based control. Both methods demonstrate improved closed-loop performance over using CBF as a separate safety filter in numerical experiments.



\end{abstract}



\IEEEpeerreviewmaketitle

\section{Introduction}
Learning-based control has become increasingly popular for controlling complex dynamical systems~\cite{pierson2017deep} since it requires little expert knowledge and can be carried out in an automatic, data-driven manner. However, due to the black-box nature of learning models, learning-based controllers such as neural network (NN) controllers lack formal guarantees which significantly limits their deployment in safety-critical applications. 

The integration of control-theoretical approaches and machine learning has provided a promising solution to safe learning control, where trainable machine learning modules are embedded into a control framework that guarantees the safety or stability of the dynamical system~\cite{donti2020enforcing, furieri2022distributed}. Of wide applicability is the control barrier function (CBF) framework~\cite{ames2016control, ames2019control} which explicitly specifies a safe control set and guards the system inside a safe invariant set. This is achieved by constructing a CBF-based safety filter that projects any reference control input (possibly generated by a NN controller) onto the safe control set online. When a continuous-time control-affine system is considered, such projection reduces to a convex quadratic program (QP) which is referred to as CBF-QP. Due to its simplicity, flexibility, and formal safety guarantees, CBFs have been applied in safe learning control with many successful applications~\cite{anand2021safe, dawson2022safe}.

\begin{figure}[!t]
\centering
	\begin{subfigure}{0.66 \columnwidth}
		\centering 
		\includegraphics[width = \textwidth]{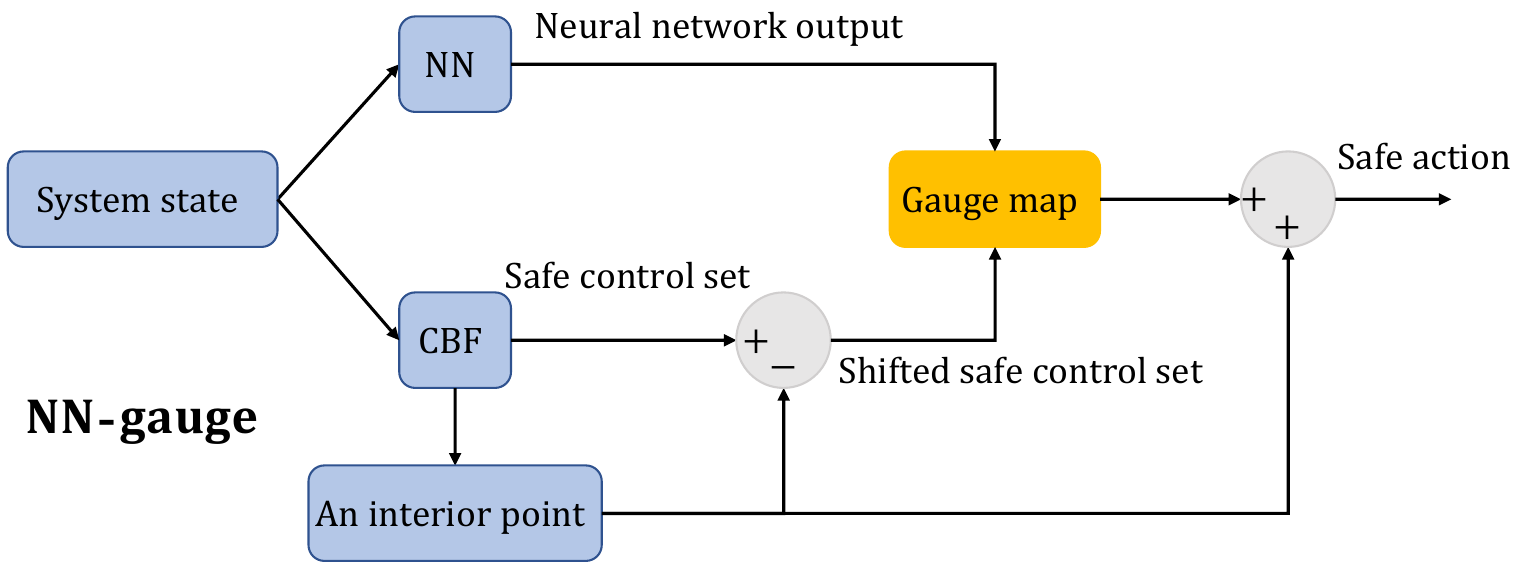}
		\caption{Gauge map-based safe NN controller architecture.}
  \label{fig:gauge_nn}
	\end{subfigure}
\hfil
	\begin{subfigure}{0.6 \columnwidth}
		\centering
		\includegraphics[width = \textwidth]{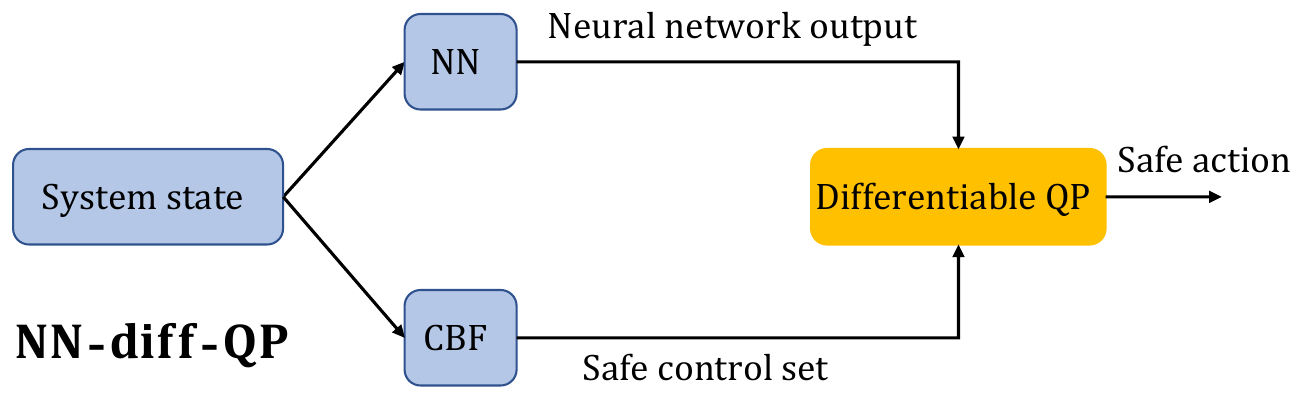}
		\caption{CBF-QP-based safe NN controller architecture.}
  \label{fig:diff_qp_framework}
	\end{subfigure}
\caption{Safe-by-construction NN controllers that utilize CBFs to construct differentiable safety layers (yellow blocks). }
\vspace{-10pt}
\label{fig:architecture}
\end{figure}

Compared with model predictive control (MPC)~\cite{camacho2013model}, which needs to solve a nonconvex optimization problem in the face of nonlinear dynamical systems, CBF-QP is computationally efficient to solve online. However, unlike MPC, the QP-based safety filter only operates in a minimally invasive manner, i.e., it generates the safe control input closest (in the Euclidean norm) to the reference control input, unaware of the long-term effects of its action. This indicates that the effects of the safety filter on the performance of the closed-loop system are hard to predict. Therefore, the application of the safety filter may give rise to myopic controllers~\cite{cohen2020approximate} that induce sub-par performance in the long term. 

To address the issue of myopic CBF-based safety filters, in this work we propose to utilize CBF to construct safe-by-construction NN controllers that allow end-to-end learning. Incorporating safety layers in the NN controller allows the learning agent to take the effects of safety filters into account during training in order to maximize long-term performance. Inspired by~\cite{tabas2022computationally}, we design a differentiable safety layer using the gauge map which establishes a bijection between the polytopic output set of a NN (e.g., an $\ell_\infty$ norm ball) and the safe control set characterized by CBFs. We denote the proposed architecture as NN-gauge (Fig.~\ref{fig:gauge_nn}). We compare NN-gauge with an alternative safe NN controller, NN-diff-QP, which consists of a NN followed by a differentiable CBF-QP layer (Fig.~\ref{fig:diff_qp_framework}). In the online execution, NN-gauge requires closed-form evaluation or solving a linear program (LP) while NN-diff-QP solves a quadratic program. Both methods are significantly cheaper to run online than MPC. 


\subsection{Related works}
\textbf{Safe controller design}: The use of gauge map in safe learning control was proposed in~\cite{tabas2022computationally, tabas2022safe} which only consider linear dynamics for which control invariant sets and an interior safe control policy are achievable. In this work, by proposing NN-gauge, we significantly extend the scope of this framework to handle nonlinear dynamics with CBFs and overcome the arising computational difficulties by applying an implicit interior policy parameterization. The construction of NN-diff-QP naturally follows from the use of differentiable optimization layers~\cite{amos2017optnet, amos2018differentiable, agrawal2019differentiable}. NN-diff-QP is applied in~\cite{pereira2021safe} and~\cite{emam2021safe} for reinforcement learning tasks. {In~\cite{cortez2022differentiable}, CBFs are applied as penalty functions to promote the safety of NN controllers} However, unlike NN-gauge or NN-diff-QP, the resulting NN controllers do not have formal safety guarantees.

\textbf{Differential CBF}: Introducing learning modules in the parameterization of CBFs can improve the feasibility and performance of the safety filter~\cite{parwana2022recursive, xu2022feasibility} even in the face of changing environments~\cite{ma2022learning, xiao2021barriernet}. These works focus on learning or improving CBFs such that the safe control set is enlarged and the safety filter can work better with a reference controller. Instead, in this paper, we consider NN controller synthesis with a given CBF. 


The contributions of the paper are summarized below:

\begin{enumerate}[leftmargin=1.5em]
    \item We propose a novel, differentiable, safe-by-construction NN controller NN-gauge as shown in Fig.~\ref{fig:gauge_nn} using CBFs. To the best of our knowledge, NN-gauge is the only alternative to the projection-based safe NN controller NN-diff-QP (Fig.~\ref{fig:diff_qp_framework}). Compared with NN-diff-QP, NN-gauge enjoys more efficient training and online evaluation.
    \item We provide detailed case studies to evaluate the performance-complexity trade-off of NN-gauge and NN-diff-QP. 
    \item We demonstrate that learning safe-by-construction NN controllers leads to better long-term closed-loop performance than filtering a trained, possibly unsafe NN controller. 
\end{enumerate}

In Section~\ref{sec:preliminary}, background on CBF is provided and the control problem is introduced. The construction of NN-gauge and NN-diff-QP is shown in Section~\ref{sec:controllerdesign}, followed by numerical examples in Section~\ref{sec:simulation}. Section~\ref{sec:conclude} concludes the paper.

\section{Preliminary and problem formulation}\label{sec:preliminary}
\subsection{System model}
In this paper we are interested in a continuous-time nonlinear control-affine system: 
\begin{align}\label{eq:system}
    \Dot{x} = f(x) + g(x)u
\end{align}
where $f$ and $g$ are locally Lipschitz, $x\in D \subseteq \mathbb{R}^n$ is the state, $D$ denotes a compact set in $\mathbb{R}^n$, and $u \in \mathbb{R}^m$ is the control input subject to bounded polytopic control constraints:
\begin{equation}\label{eq:input_constr}
   u \in  U = \{u \in \mathbb{R}^m \mid A_u u \leq b_u\}.
\end{equation}

\subsection{Control barrier functions}
Safety can be framed in the context of enforcing set invariance in the state space, i.e., the state cannot exit a safe set. The safe set $\mathcal{C}$ is represented by the superlevel set
of a continuously differentiable function $h(x)$. The algebraic expressions of the safe
set $\mathcal{C}$, boundary of the safe set $\partial \mathcal{C}$, and interior of the safe set $\text{Int}(\mathcal{C})$ are given by:
\begin{align}
 \mathcal{C}&=\{x\in \mathbb{R}^n:h(x)\ge 0\}, \nonumber\\
 \partial\mathcal{C}&=\{x\in D\subset \mathbb{R}^n:h(x)= 0\}, \\
 \text{Int}(\mathcal{C})&=\{x\in D\subset \mathbb{R}^n:h(x)> 0\}. \nonumber
\end{align}
For a locally Lipschitz continuous control law $u=k(x)$, we have that $\Dot{x}=f(x)+g(x)k(x)$ is locally Lipschitz continuous. Thus, for any initial condition $x_0\in D$, there exists a maximum time interval of existence $I(x_0)=[0,\tau_{max})$, such that $x(t)$ is the unique solution to the ordinary differential equation~(\ref{eq:system}) on $I(x_0)$. We frame the safety of system~\eqref{eq:system} in terms of set invariance as shown below. 
\begin{mydef}(Forward invariance and safety)
The set $\mathcal{C}$ is \emph{forward invariant} if for every $x_0\in \mathcal{C}$, $x(t)\in \mathcal{C}$ holds for $x(0)=x_0$ and all $t\in I(x_0)$. If $\mathcal{C}$ is forward invariant, we say system~\eqref{eq:system} is safe. 
\end{mydef}
To verify invariance of $\mathcal{C}$ under the control input constraints~\eqref{eq:input_constr}, a control barrier function is constructed as a certificate which characterizes the admissible set of control inputs that render $\mathcal{C}$ forward invariant. 
%
\begin{mydef}(Extended class $\mathcal{K}_{\infty}$ function)
A continuous function $\alpha$: $\mathbb{R}\rightarrow \mathbb{R}$ is said to be an extended $\mathcal{K}_{\infty}$ if it is strictly increasing and $\alpha(0)=0$.
\end{mydef}
\begin{mydef}(Control barrier function)
Let $\mathcal{C}\subset D\subset \mathbb{R}^n$ be the superlevel set of a continuously differentiable function $h: D\rightarrow \mathbb{R}$, then $h$ is
a control barrier function if there exists an extended class $\mathcal{K}_{\infty}$ function $\alpha(\cdot)$ such that for the control system (\ref{eq:system}):
\begin{align}
    \mathsf{sup}_{u\in U}[L_fh(x)+L_gh(x)u]\ge -\alpha(h(x)),
\end{align}
for all $x\in D$ where $L_f h(x)$ and $L_g h(x)$ denote the Lie derivatives. 
\end{mydef}
Given the CBF $h(x)$, the set of all control values that render $\mathcal{C}$ safe is given by:
\begin{equation}\label{eq:safe_control_set}
    K_{cbf}(x) = \{u\in U: L_fh(x)+L_gh(x)u+\alpha (h(x))\ge 0\}
\end{equation}
which we denote as the safe control set. The following theorem shows that the existence of a control barrier function implies that the control system~\eqref{eq:system} is safe:
\begin{mythm}(\cite[Theorem 2]{ames2016control})
\label{thm:cbf}
Assume $h(x)$ is a CBF on $D \supset \mathcal{C}$ and $\frac{\partial h}{\partial x}(x) \neq 0$ for all $x \in \partial \mathcal{C}$. Then any Lipschitz continuous controller $u(x)$ such that $u(x) \in K_{cbf}(x)$ for all $x \in \mathcal{C}$ will render the set $\mathcal{C}$ forward invariant.
\end{mythm}

One important feature of $K_{cbf}(x)$ is that it is a polytope~\footnote{The polytopic safe control set $K_{cbf}(x)$ is non-empty for all $x \in \mathcal{C}$ by definition of CBF. In this work, we assume a valid CBF for system~\eqref{eq:system} is given and use it for NN controller design, although synthesizing CBFs can be challenging itself and is an active area of research. } for all states. This enables the construction of a QP-based safety filter that modifies any given reference controller $k(x)$ in a minimally invasive fashion~\cite{ames2019control} as follows:
\begin{equation} \label{eq:safety_filter}
\begin{aligned}
    u(x) = \underset{u}{\text{argmin}} & \quad \lVert u - k(x) \rVert_2^2\\
    \text{subject to} & \quad L_fh(x)+L_gh(x)u+\alpha (h(x))\ge 0, \\
    & \quad u \in U.
\end{aligned}
\end{equation}
%

Although the safety filter~\eqref{eq:safety_filter} guarantees the forward invariance of the safe set, it does not take into account the consequences of the projection on future states and the performance of the closed-loop system. This issue is inevitable when the reference controller $k(x)$ and the safety filter are designed separately, and we propose to fix it by designing safe-by-construction controllers that are amenable to any learning or optimization framework. Particularly, in this work, we consider optimizing NN controllers using modern machine learning solvers (such as stochastic gradient descent (SGD) and Adam~\cite{adam}) with known system dynamics. 



\subsection{Problem formulation}
Following the definition of safe control set~\eqref{eq:safe_control_set}, we define the set of safe control policies as
$\Pi :=\{\pi: \mathbb{R}^n \mapsto \mathbb{R}^m|\pi(x) \in K_{cbf}(x), \forall x\in \mathcal{C}\}.$
Our task is to design a controller for system~\eqref{eq:system} such that a performance objective is optimized and the closed-loop system always stays inside the safe set $\mathcal{C}$. In other words, we consider finding a policy $\pi(x)$ that solves the following optimal control problem within a horizon $T < \tau_{max}$~\footnotemark:
\footnotetext{The horizon $T$ is a tuning parameter for NN controller design. While a larger $T$ is always preferred to improve the closed-loop performance of the trained NN controller, it necessarily increases the computational complexity of training.}
\begin{equation}\label{eq:control_prob}
\begin{aligned}
 \underset{\pi \in \Pi}{\text{minimize}} & \quad \mathbb{E}_{x(0)\in \mathcal{C}} \Big [ \frac{1}{T}\int_{0}^{T} c(x(t), u(t)) \,dt \Big ]\\
    \text{subject to} & \quad \Dot{x}=f(x)+g(x)u,\\
    & \quad u(t)=\pi(x(t)),
\end{aligned}
\end{equation}
where $\mathbb{E}_{x(0)\in \mathcal{C}}$ is the expectation with respect to the initial state, $c(x(t), u(t))$ is the cost associated with occupying state $x(t)$ and action $u(t)$. The cost function $c(x(t), u(t))$ could be of any form and is problem-specific, e.g., it can formulate the penalty or barrier functions of constraints on the state $x(t)$.

Despite the complex dynamics, cost functions, and safe policy constraints, problem~\eqref{eq:control_prob} can still be effectively approached by parameterizing an NN policy and applying SGD/Adam which is empowered by automatic differentiation and modern machine learning solvers~\cite{baydin2018automatic}. This procedure is simple and has been shown effective to synthesize high-performance NN controllers~\cite{drgona2020learning, mukherjee2022neural}. Next, we study different parameterizations of safe NN policies $\pi \in \Pi$ for solving~\eqref{eq:control_prob} and evaluate their performances. Notably, with safe NN policies, the effects of the CBF-based safety layer are automatically considered.

\section{Safe Controller Design}
\label{sec:controllerdesign}
A natural way to construct a safe NN controller $\pi(x)$ is to restrict the output of the controller into the safe control set~\eqref{eq:safe_control_set} for all states. NN-diff-QP (Fig.~\ref{fig:diff_qp_framework}) achieves this by concatenating a differentiable projection layer which can be implemented in a NN using toolboxes such as \texttt{cvxpylayers}~\cite{agrawal2019differentiable} and \texttt{qpth}~\cite{amos2017optnet}. In this section, we propose a different parameterization of a safe NN controller which achieves improved online computational efficiency using gauge maps. 

\subsection{Gauge map}
Tabas et al.~\cite{tabas2022computationally} observe that, while it is challenging to directly restrict the output of a NN inside a general polytope, adding a hyperbolic tangent activation layer easily constrains the NN output into the unit $\ell_\infty$-norm ball $\mathbb{B}_{\infty}$. This motivates the application of the gauge map that establishes a bijection between $\mathbb{B}_{\infty}$ and a general polytope which in this paper we consider as $K_{cbf}(x)$. The notion of gauge map is facilitated by the concept of C-set.

\begin{mydef} (C-set~\cite{blanchini2008set})
A C-set is a convex and compact set including the origin as an interior point.
\end{mydef}

The gauge function (or Minkowski function) of a vector $v\in \mathbb{R}^m$ with respect to a C-set $Q\subset \mathbb{R}^m$ is given by
\begin{align}
    \gamma_{Q}(v)=\mathsf{inf}\{\lambda\ge 0|v\in \lambda Q\}.
\end{align}
When $Q$ is a polytopic C-set defined by $\{w\in\mathbb{R}^m|F^T_iw\le g_i, i=1,2,\cdots, r\}$, the gauge function can be written in closed-form as
$
    \gamma_{Q}(v)=\mathsf{max}_i\{F_i^Tv/g_i\}.
$
For any $v\in \mathbb{B}_{\infty}$, the gauge map is defined as
\begin{align}\label{gaugemap}
    G(v|Q)=\frac{||v||_{\infty}}{\gamma_{Q}(v)}\cdot v,
\end{align}
which constructs a bijection (see \cite[Lemma 1]{tabas2022computationally}) between the unit ball $\mathbb{B}_{\infty}$ and the C-set $Q$. As shown in Fig.~\ref{fig:gauge_illustration}, all points $v \in \mathbb{B}_\infty$ are mapped ``proportionally" to the same level set of the polytope $Q$ by the gauge map, while their projections onto $Q$ tend to concentrate on the boundary of $Q$.  

\begin{figure}[t]
\includegraphics[width=8.5cm]{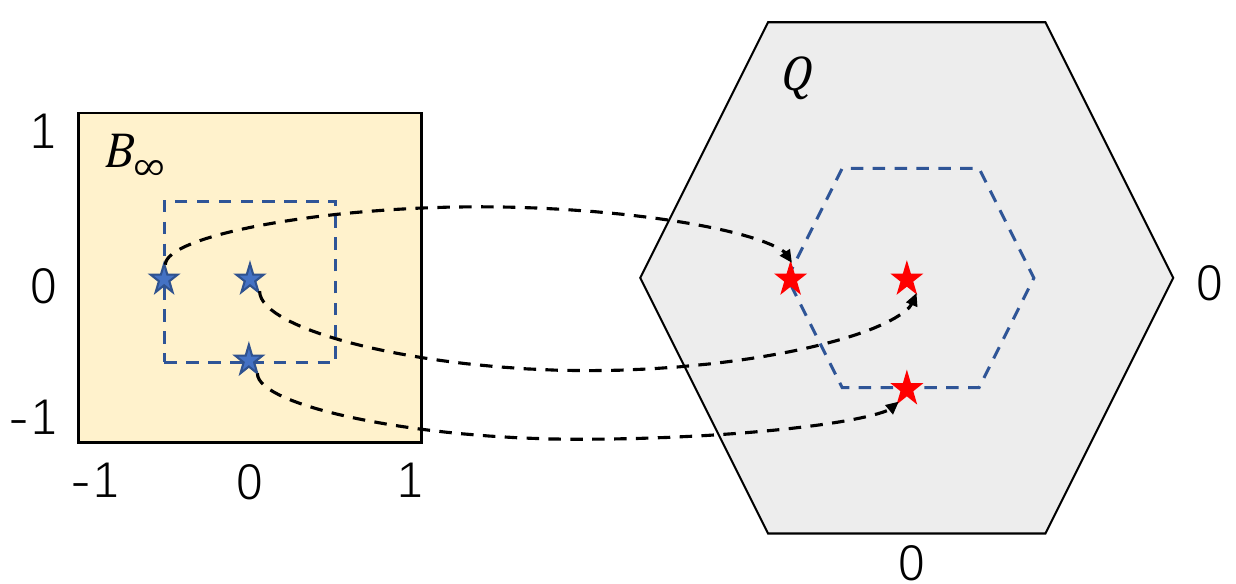}
\centering
\caption{Illustration of the gauge map from the $\ell_\infty$ ball $B_\infty$ to a polytopic set $Q$.}
\label{fig:gauge_illustration}
\end{figure}




\subsection{Interior safe policy}
To apply the gauge map as a safety layer that maps the output of an NN to the safe control set $K_{cbf}(x)$, we have to first find an interior safe policy $\pi_{int}(x)$ such that $\pi_{int}(x) \in \text{Int}(\mathcal{C})$ for all $x \in \mathcal{C}$. This is necessary since $K_{cbf}(x)$ may not be a C-set, and we need to shift $K_{cbf}(x)$ by $\pi_{int}(x)$ such that it is recentered around the origin. When $\pi_{int}(x)$ is available, we have that the shifted safe control set $\hat{K}_{cbf}(x) :=K_{cbf}(x) - \pi_{int}(x)$ is a C-set for which we can apply the gauge map. 

An explicit construction of $\pi_{int}(x)$ is achievable for linear dynamical systems through multi-parametric programming~\cite{tabas2022safe}. \change{The recent work~\cite{wang2022explicit} proposes an algorithm to extract an explicit or closed-form safe policy from CBFs for system~\eqref{eq:system}. Such an explicit interior policy is desirable since in this case the gauge map can be evaluated in closed-form, making both the training and online evaluation of NN-gauge computationally efficient. However, for general nonlinear systems, the proposed algorithm in~\cite{wang2022explicit} can be complex to apply in practice. To address this issue, we propose an alternative method that implicitly constructs an interior safe policy by choosing $\pi_{int}(x)$ as the Chebyshev center~\cite{boyd2004convex} of $K_{cbf}(x)$.} Specifically, we have $\pi_{int}(x) = u^*$ where $u^*$ is the solution to the following LP:
%
\begin{equation} \label{eq:chebyshev}
\begin{aligned}
    \underset{u, R \geq 0}{\text{maximize}} &\quad R \\
    \mathsf{subject \:to} & \quad -L_gh(x)u+R\lVert -L_gh(x)\rVert_2\\
    &\quad\quad\quad\quad\quad\quad\quad \le L_fh(x)+\alpha(h(x)), \\
    & \quad a_{u, i}^\top u + R \lVert a_{u,i} \rVert_2 \leq b_{u, i}, \ i = 1, \cdots, d_u,\\
\end{aligned}
\end{equation}
where $a_{u, i}$ denotes the $i$-th row of $A_u$, $b_{u,i}$ denotes the $i$-th entry of $b_u$, and $d_u$ represents the number of linear constraints defining $U$. The Chebyshev center formulation of $\pi_{int}(x)$ pushes $\pi_{int}(x)$ away from the boundary of $K_{cbf}(x)$. It also facilitates training of the upstream NN, since it makes the target function that the NN needs to learn smoother. By the validity of CBF, the safe control set $K_{cbf}(x)$ is a non-empty polytope for all $x \in \mathcal{C}$. By choosing $\pi_{int}(x)$ as the Chebyshev center from~\eqref{eq:chebyshev}, we readily have that $\pi_{int}(x) \in \text{Int}(K_{cbf}(x))$ for $R> 0$ and the shifted safe control set $\hat{K}_{cbf}(x)$ is a C-set.






\subsection{Control policy architecture}
With an interior safe policy $\pi_{int}(x)$, we can now construct a safe NN controller using the gauge map, as shown in the following theorem. 
\begin{mythm} \label{thm:gauge}
Let $\phi_{\theta}: \mathcal{C}\rightarrow\mathbb{B}_{\infty}$ be a neural network parameterized by $\theta$ and $\pi_{int}(x)$ be an interior safe policy. Then, for any system state $x$ in the set $\mathcal{C}$, the policy:
\begin{equation} \label{eq:gauge_nn}
    \pi_{\theta}(x):=G(\phi_{\theta}(x)|\hat{K}_{cbf}(x))+\pi_{int}(x)
\end{equation}
has the following properties:
\begin{enumerate}
    \item $\pi_{\theta}$ is safe.
    \item The policy $\pi_\theta(\cdot)$ is trainable with respect to the NN parameters $\theta$.
\end{enumerate}
\end{mythm}
\begin{proof}
1) By the construction of the gauge map, we have $G(\phi_{\theta}(x)|\hat{K}_{cbf}(x)) \in \hat{K}_{cbk}(x)$ which is the safe control set $K_{cbf}(x)$ shifted by the interior safe control input $\pi_{int}(x)$. Therefore, we have $\pi_\theta(x) \in K_{cbf}(x)$ for all $x \in \mathcal{C}$, and conclude that $\pi_\theta(x)$ is safe.

2) As shown in~\cite[Theorem 1]{tabas2022computationally}, automatic differentiation can be applied to compute the subgradients of $\pi_\theta$ with respect to $\theta$, making it possible to train $\pi_\theta$.
\end{proof}
The NN $\phi_\theta$ can embed any architecture and learns the residual control policy added to the interior safe policy $\pi_{int}(x)$. By the construction of the gauge map, the applied controller $\pi_\theta(x)$ shown in~\eqref{eq:gauge_nn} belongs to the safe control set $K_{cbf}(x)$, and the performance of $\pi_\theta(x)$ is no worse than $\pi_{int}(x)$ after training. The online evaluation of $\pi_\theta$ can be done in closed-form if an explicit interior safe policy $\pi_{int}(x)$ is given, or by solving an LP if the implicit construction~\eqref{eq:chebyshev} is used. 

\begin{remark}
Our analysis of safe NN controllers can be readily applied to incorporate high order CBFs~\cite{xiao2019control} in which case the safe control set $K_{cbf}(x)$ is still a polytope. In addition to the CBF-based safe control sets, NN-gauge and NN-diff-QP can easily encode other forms of polytopic safe control sets such as $u \in \{u \mid Gu \leq b\}$ and $u \in \{u \mid Fx + Gu \leq b\}$. 
\end{remark}

\section{Numerical Experiments}\label{sec:simulation}
In this section, we demonstrate the application of safe NN controllers in adaptive cruise control (ACC) and aircraft collision avoidance. The following control methods are considered:
\begin{enumerate}
    \item \textbf{MPC}: Model predictive control could guarantee the safety of the closed-loop system with good performance, but it is computationally costly to run online when the horizon is large or the dynamics is nonlinear. 
    \item \textbf{NN}: A feedforward NN controller is trained to optimizes~\eqref{eq:control_prob} with regularizers penalizing violations of safety constraints. No safety filter is applied in the online evaluation of this NN controller. This method enjoys fast online evaluation but suffers the risk of safety violations.
    \item \textbf{NN-QP}: The above NN controller is equipped with the CBF-QP safety filter during online evaluation. 
    \item \textbf{NN-diff-QP} and \textbf{NN-gauge}: The safe-by-construction NN controllers introduced in Section~\ref{sec:controllerdesign} that are trained directly to optimize the control performance~\eqref{eq:control_prob}. \change{We also include the interior policy $\pi_{int}$ given by~\eqref{eq:chebyshev} as a special case of NN-gauge for comparison.}
\end{enumerate}
For NN-diff-QP, we use \texttt{cvxpylayers}~\cite{agrawal2019differentiable} to construct the differentiable QP layer. For NN-gauge, the implicit interior policy~\eqref{eq:chebyshev} is applied. All training is performed using PyTorch on Google Colab with Adam~\cite{adam} as the optimizer.

\subsection{Adaptive cruise control}
Adaptive cruise control is a common example to validate safe control strategies~\cite{ames2016control, zeng2021safety}. The control goal of ACC is to let the ego car achieve the desired cruising speed while maintaining a safe distance from the leading car. We consider the scenario where the ego car tries to follow the leading car on a straight road. The dynamics of the problem is given by (model adapted from~\cite{zeng2021safety}):
\begin{align} \label{eq:ACC}
    \Dot{x}(t)=\begin{bmatrix}
0 & 1 & 0 \\
0 & -0.1 & 0\\
0 & -1 & 0
\end{bmatrix}x(t)+
\begin{bmatrix}
0\\
2.5\\
0 
\end{bmatrix}u(t),
\end{align}
where $x=\begin{bmatrix}
x_1 & x_2 & x_3
\end{bmatrix}^\top=\begin{bmatrix}
p_f & v_f & d
\end{bmatrix}^\top$, $p_f$ is the position, $v_f$ is the velocity of the ego car, $d$ is the distance between the ego and leading cars, and $u\in [-1, 1]$ is the control input denoting the acceleration. To prevent collision between the two cars, the CBF is chosen as $h(x)=x_3-1.8x_2$
and the trajectory cost for the ego car is given by
\begin{align}
    \texttt{Cost}=\int_{0}^T (0.01(x_2(t)-v_{f}^*)^2+0.05u(t)^2) dt
\end{align}
with a desired speed $v_{f}^*=30m/s$. The leading car travels at a constant speed of $16 m/s$, so we expect the ego car's speed to converge to $16 m/s$ at the steady state since a speed greater than $16 m/s$ will lead to a violation of safe distance.

All NN controllers, namely NN, NN-gauge and NN-diff-QP, are trained to optimize~\eqref{eq:control_prob} with horizon $T = 1$s by Adam with randomly sampled initial states. The system dynamics~\eqref{eq:ACC} is discretized with sampling rate $\Delta t = 0.1$s during training. The trainable NN modules $\phi_\theta$ in NN, NN-gauge, and NN-diff-QP have the same architecture. 

We test all NN controllers on 5 randomly sampled initial states over a horizon $T = 20$s in order to evaluate their long-term performance. The results are reported in Table~\ref{table:acc}. For one of the testing initial condition 
$x(0)=\begin{bmatrix}
0 & 30 & 100
\end{bmatrix}^\top$,
we plot the values of the CBF function and the velocity along the trajectory of the ego car in Fig.~\ref{fig:accVelocity} and Fig.~\ref{fig:accCBFvalue}, respectively.

We observe that both NN-gauge and NN-diff-QP achieve similar closed-loop performance, comparable to MPC~\footnote{While MPC solves a finite horizon optimal control problem to optimality, its closed-loop performance is not guaranteed to be optimal. In Table~\ref{table:acc}, NN-diff-QP achieved better performance than MPC in the long term.}. While NN achieves reasonable performance, it violates safety constraints (Fig.~\ref{fig:accCBFvalue}). Directly applying the CBF-QP safety filter on it enforces safe control, but deteriorates the long-term closed-loop performance of the NN controller as shown in Fig.~\ref{fig:accVelocity} where the optimal behavior of the ego car is supposed to have a steady-state velocity of $16 m/s$. NN-gauge has an edge over NN-diff-QP in training and online evaluation time due to its use of LP-based safety layers. \change{We also observe a large performance improvement of NN-gauge compared with $\pi_{int}$ which means training the NN module in NN-Gauge can greatly improve the performance of the policy.}

\begin{table}
\centering
\small
\caption{Closed-loop performance comparison of controllers in the adaptive cruise control example. Average values from $5$ randomly sampled initial states are reported.}
\label{table:acc}
\begin{tabular}{*{5}{c}}
\toprule
 & Safety & Trajectory & Training time& Solve\\
 &  & cost & per epoch& time\\
\midrule
MPC & Safe & 269.3 &N/A &3.11s \\
NN & Unsafe & 640.3 &0.03s &0.04s \\
NN-QP & Safe & 818.2 &0.03s &0.36s \\
NN-diff-QP & Safe & \textbf{258.3}&11.0s & \textbf{0.35s} \\
NN-gauge & Safe & \textbf{270.9} &8.3s &\textbf{0.28s} \\
\cdashlinelr{1-5}
 $\pi_{int}$ & Safe    & 734.8 & N/A &0.26s\\
\bottomrule
\end{tabular}
\end{table}

\begin{figure}[!t]
\centering
	\begin{subfigure}{0.38 \columnwidth}
    \centering
    \includegraphics[width=\textwidth]{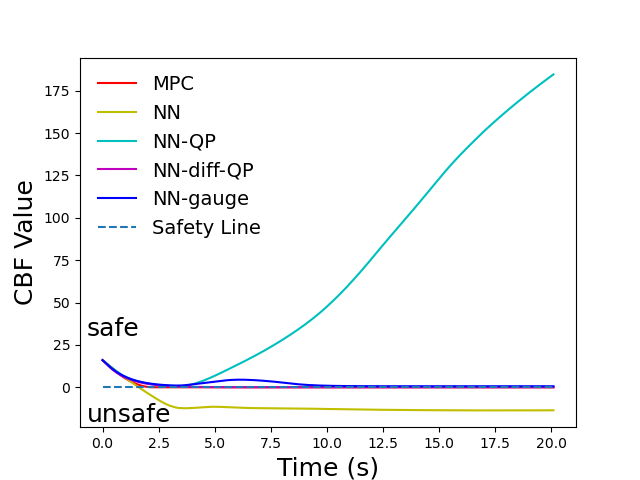}
    \caption{CBF values.}
    \label{fig:accCBFvalue}
	\end{subfigure}
		\begin{subfigure}{0.38 \columnwidth}
    \centering
    \includegraphics[width=\textwidth]{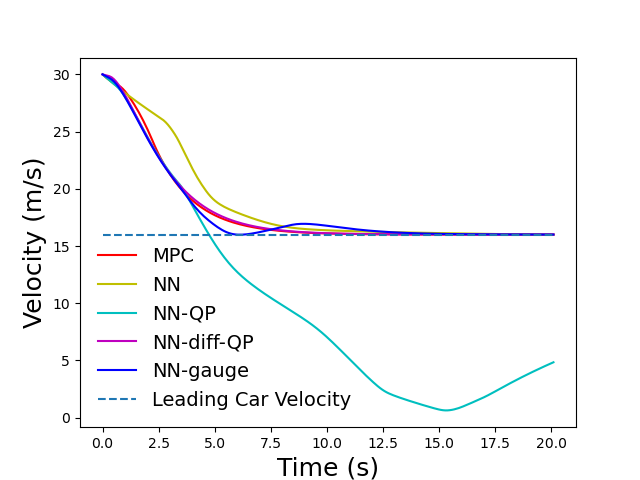}
    \caption{Velocity of the ego car.}
    \label{fig:accVelocity}
	\end{subfigure}
\caption{CBF values (left) and velocity of the ego car (right) under different controllers are evaluated in closed-loop for $20$s. A CBF value below zero indicates unsafety, and the optimal behavior of the ego car is expected to have a steady state velocity of $16 m/s$ same as the leading car. }
\label{fig:acc}
\end{figure}


\subsection{Aircraft collision avoidance}
We apply our framework to the aircraft collision avoidance problem which is adapted from~\cite{squires2018constructive}. Specifically, we consider a dynamical system with states
$x=\begin{bmatrix}
x_a^\top & x_b^\top
\end{bmatrix}^\top$, where
$x_a=
\begin{bmatrix}
p_{a, x} & p_{a, y} & \theta_a
\end{bmatrix}
^\top\in \mathbb{R}^3$ is the state of aircraft $a$ with $p_{a, x}$ and $p_{a, y}$ denoting the position and $\theta_a$ denoting the orientation. The state $x_b$ of aircraft $b$ is defined similarly. The control inputs are $u= \begin{bmatrix}
v_a & w_a & v_b & w_b
\end{bmatrix}
^\top\in \mathbb{R}^4$ where $v_a$ and $w_a$ are speed and turning rate of aircraft $a$, respectively, and $v_b, w_b$ are defined similarly. The dynamics of the aircraft $a$ (and similarly for aircraft $b$) is given by:
\begin{align}
    \Dot{x}_a(t)=\begin{bmatrix}
v_a(t)\text{cos}(\theta_a(t)) \\
v_a(t)\text{sin}(\theta_a(t))\\
w_a(t)
\end{bmatrix}.
\end{align}

As shown in Fig.~\ref{fig:aircraftDot}, our goal is to drive aircraft $a$ to the left and aircraft $b$ to the right while maintaining a minimum safe distance of $0.5$ between them. Aircrafts $a$ and $b$ try to stay close to $p_{a,x} = -5$ and $p_{b,x} = 5$, respectively. A quadratic cost function is defined accordingly over a horizon $T$, and we adopt the constructive CBF developed in~\cite{squires2018constructive} to encode the safe set in the state space which also considers the input constraint $ (v_a, w_a, v_b, w_b)\in [0.1, 1]\times [-1, 1]\times [0.1, 1]\times [-1, 1].$
Note that the minimum admissible velocity of aircraft $a$ and $b$ is $0.1$, so they cannot stop exactly at $p_{a,x} = -5$ or $p_{b,x} = 5$.



\begin{table}
\centering
\small
\caption{Closed-loop performance comparison of controllers in the aircraft collision avoidance example. Average values from $5$ randomly sampled initial states are reported.}
\label{table:aircrafts}
\begin{tabular}{*{5}{c}}
\toprule
 & Safety & Trajectory & Training time& Solve\\
 &  & cost & per epoch& time\\
\midrule
MPC & Safe & 1346.8 &N/A &44.1s \\
NN & Unsafe & 712.0 &0.24s &0.06s \\
NN-QP & Safe & 2557.6 &0.24s &1.28s \\
NN-diff-QP & Safe & \textbf{1900.7}&133.0s & \textbf{1.30s} \\
NN-gauge & Safe & \textbf{2157.8} &101.5s& \textbf{0.74s} \\
\cdashlinelr{1-5}
 $\pi_{int}$ & Safe    & 3495.9 & N/A &0.69s\\
\bottomrule
\end{tabular}
\end{table}

All NN controllers are trained similarly as in the ACC example with sampling rate $\Delta_t = 0.1$s and horizon $T =2$s, and are tested on 5 randomly sampled initial states with horizon $T = 20$s. The results are shown in Table~\ref{table:aircrafts}. One set of closed-loop system trajectories starting from the initial state $x(0)=\begin{bmatrix}
0.5 & 0 & \pi & -0.5 & 0 & 0
\end{bmatrix}^\top$ under different controllers are plotted in Fig.~\ref{fig:aircraftDot} together with the induced costs. With this initial condition, aircrafts $a$ and $b$ start close to each other with orientations leading to a head-on collision. 

\begin{figure}
    \centering
    \includegraphics[width=0.5\columnwidth]{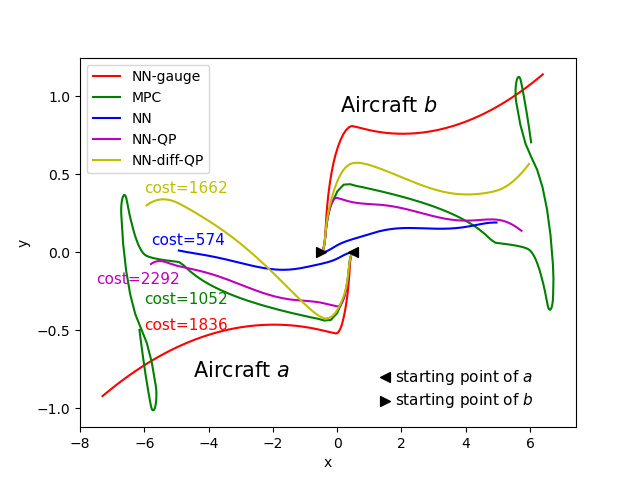}
    \caption{Trajectories of the aircraft $a$ and $b$ under different controllers. The induced costs of each trajectory are labeled accordingly. }
    \label{fig:aircraftDot}
\end{figure}

From Table~\ref{table:aircrafts} and Fig.~\ref{fig:aircraftDot}, we observe that NN achieves the best performance but is unsafe. Adding the CBF-QP as a safety filter (i.e., NN-QP) drastically deteriorates the performance of the NN controller. Among the NN controllers with safety guarantee, NN-diff-QP performs the best and NN-gauge achieves a similar level of performance. The MPC controller has the best performance with safety guarantee, but it has a significantly higher online solve time. 
\section{Conclusion}\label{sec:conclude}
In this paper, we showed that CBF-based safety filters can degrade closed-loop performance if their long-term effects are not considered during learning. To address this issue, we proposed a novel safe-by-construction NN controller which utilizes CBF and gauge map to construct a differentiable safety layer. The proposed gauge map-based NN controller achieves comparable performances as the projection-based NN controller while being computationally more efficient to train and evaluate online. Both the gauge map-based and projection-based safe NN controllers demonstrate improved performance compared with filtered NN controllers in numerical examples.  

\section*{Acknowledgement}
We thank Nikolai Matni for useful discussions.
This project is funded in part by the US Department of Transportation's Mobility21 National University Transportation Center and NSF CCRI $\#1925587$.

\bibliographystyle{unsrt}
\bibliography{reference}


\end{document}